\documentclass[11pt]{article}

\usepackage{fullpage, xcolor, amsmath, amssymb, amsthm, authblk}
\usepackage[pdfborder={0 0 0.5}]{hyperref}

\newcommand{\ket}[1]{|#1\rangle}
\newcommand{\braket}[2]{\langle #1|#2\rangle}

\newcommand{\id}{\mathrm{1}}

\newtheorem{lemma}{Lemma}

\date{}

\title{Efficient quantum processing of ideals in finite rings}

\author[1]{Pawel~M.~Wocjan}
\author[2,4]{Stephen~P.~Jordan}
\author[3]{Hamed~Ahmadi}
\author[3]{Joseph~P.~Brennan}
\affil[1]{\small{IBM Quantum, Thomas J Watson Research Center, Yorktown Heights, NY\footnote{pawel.wocjan@ibm.com}}}
\affil[2]{\small{Google Quantum AI, Seattle, WA\footnote{stephenjordan@google.com}}}
\affil[3]{\small{Department of Mathematics, University of Central Florida, Orlando}}
\affil[4]{\small{UMIACS, University of Maryland, College Park, MD}}
\begin{document}

\maketitle

\begin{abstract}
Suppose we are given black-box access to a finite ring $R$ that is not necessarily commutative, and a list of generators for an ideal $I$ in $R$. We show how to find an additive basis representation for $I$ in $\mathrm{poly}(\log|R|)$ quantum time. This generalizes a quantum algorithm of Arvind \emph{et al.} which finds a basis representation for $R$ itself. We then show this primitive allows quantum computers to rapidly solve a wide variety of problems regarding finite rings. In particular we show how to test whether two ideals are identical, find their intersection, find their quotient, prove whether a given ring element belongs to a given ideal, prove whether a given element is a unit, and if so find its inverse, find the additive and multiplicative identities, compute the order of an ideal, solve linear equations over rings, decide whether an ideal is maximal, find annihilators, and test the injectivity and surjectivity of ring homomorphisms. These problems appear to be hard classically.
\end{abstract}

\section{Introduction}

A non-commutative ring has both non-commutative monoid structure with respect to multiplication and Abelian group structure with respect to addition, and both operations are compatible in the sense that they satisfy the distributive law. Here we consider the application of quantum techniques to solve problems regarding ideals in finite rings, which need not be commutative. The ideals are specified by a list ring elements that generate the ideal under addition and multiplication by arbitrary ring elements. Despite the involvement of non-commutative multiplication, these problems can be solved by leveraging efficient quantum algorithms for the Abelian hidden subgroup problem. The first step in these algorithms is to find a set of ring elements that generates the ideal under addition only. This generalizes a previous quantum algorithm by Arvind \emph{et. al} that finds an additive generating set for the whole ring and thereby solves some classically hard ring problems~\cite{Arvind}.

In the main body of this paper we consider various problems regarding ideals in non-commutative rings and present quantum reductions to Abelian hidden subgroup problems. These reductions can be built in natural and simple ways from fundamental quantum primitives such as the Hadamard test. From a complexity-theoretic point of view it is worth noting that these problems can also be reduced to the Abelian hidden subgroup problem using purely classical polynomial-time reductions. Efficient classical algorithms for solving systems of linear Diophantine equations power the reductions by which problems can be solved using these black boxes. We illustrate this alternative approach in section \ref{sec:alt}.

As shown in \cite{Kayal}, both integer factorization and graph isomorphism reduce to the problem of counting automorphisms of
rings. This counting problem is contained in AM$\cap$coAM. Therefore it is unlikely to be NP-hard. Integer factorization also reduces to the problem of finding nontrivial automorphisms of rings and to the problem of finding isomorphisms between two rings. Furthermore, graph ismorphism  reduces to ring isomorphism for commutative rings. Thus these ring automorphism and isomorphism problems are attractive targets for quantum computation. Perhaps the quantum algorithms given in this paper can serve as steps toward efficient quantum algorithms for some of these problems.

Let $R$ be a finite ring with identity, which need not be commutative. Let $\tilde{R} = \{r_1,\ldots,r_n\}$ be a subset of $R$ such that each element of $R$ can be obtained by some sequence of additions and multiplications of elements of $\tilde{R}$. We say that $\tilde{R}$ is a generating set for $R$. Let $I$ be the left ideal in $R$ generated by $\tilde{I}$. That is, $I$ is the smallest subset of $R$ containing $\tilde{I}$ that is closed under addition and closed under left multiplication by elements of $R$. Throughout this paper we mainly discuss left ideals. One can similarly define right ideals and two-sided ideals, and the generalization of our algorithms to these cases is a straightforward generalization. Note that $R$ is itself an ideal in $R$.

A left ideal $I$ in a finite ring $R$ forms an Abelian group $(I,+)$ under addition. Any generating set $\{a_1,\ldots,a_\ell\}$ for an Abelian group $A$ yields a homomorphism from $\mathbb{Z}_{s_1} \times \ldots \times \mathbb{Z}_{s_\ell}$ to $A$ where $s_1,\ldots,s_l$ are the orders of $a_1,\ldots,a_\ell$. In additive notation, this homomorphism takes the integers $z_1,\ldots,z_\ell$ to $\sum_{j=1}^\ell z_j a_j$. The structure theorem for finite Abelian groups states that there exists a generating set for $A$ such that this homomorphism is an isomorphism. We call this a {\em generating set of the invariant  factors}, or i.~f. generating set for short. The first step in the algorithms given  in this paper is an efficient quantum algorithm to find an i.~f. generating set for $(I,+)$. This works similarly to the efficient quantum algorithm of \cite{Arvind} for finding a basis representation of $R$. No polynomial time classical algorithm for these problems are known.

The computational difficulty of problems on rings may depend on how the algorithm is allowed to access the ring. We assume only black box access to the ring. That is, the ring elements are assigned arbitrary bit strings by some injective map $\eta$ and we have access to black boxes implementing $f_+(\eta(a),\eta(b)) = \eta(a+b)$ and $f_\times(\eta(a),\eta(b)) = \eta(a \times b)$. The ideal $I$ is specified by a list of generators $\tilde{I} = \{i_1,\ldots,i_m\}$ with $m = O(\log |R|)$. Given these inputs, our method for finding an i.~f. generating set for $(I,+)$ proceeds in two steps. First we find a generating set for $(I,+)$. Although the elements of $\tilde{I}$ generate $I$ as an ideal, they do not generate $I$ as an Abelian group, that is, by addition only with no left-multiplication by $R$ elements. After finding a generating set for $(I,+)$ we then convert it to an i.~f. generating set for $(I,+)$ using the quantum algorithms of \cite{Cheung_Mosca, Watrous}.

\section{Finding a generating set of invariant factors}
\label{sec:if}

To find a generating set for $(I,+)$, let $\tilde{B}_1 = \tilde{I}$ and apply the following iteration. Let $B_k$ be the Abelian group additively generated by $\tilde{B}_k$. At the $k^{\mathrm{th}}$ step we search for an element $i \in I$ not contained in $B_k$. If we find one, we let $\tilde{B}_{k+1} = \tilde{B}_k \cup \{i\}$. For some sufficiently large $k$, $B_k = I$, at which point the search for $i$ fails and the process terminates. We now show in detail how this works and that we need at most $\log_2|I|$ iterations.

Suppose we know $\tilde{B}_k$. To find an element of $I$ not contained in $B_k$, we check for each $r\in\tilde{R}$ and each $b\in\tilde{B}_k$ if $rb\in B_k$. To this end, we create the superposition 
\[
  \ket{B_k} = \frac{1}{\sqrt{|B_k|}} \sum_{x \in B_k} \ket{x}
\]
and use the Hadamard test to check if it is invariant under the translation $\ket{x} \mapsto \ket{rb + x}$.  If $rb\in B_k$, the state is invariant, and we always measure $1$.  Otherwise the state is mapped to an orthogonal state, and we measure $0$ or $1$ with equal probabilities.  Repeating the Hadamard test multiple times, we can determine if $rb\in B_{k}$ with exponentially small error probability. Because $B_k$ is an Abelian group whose generators we know, $\ket{B_k}$ can be created efficiently to polynomial precision using the results of \cite{Cheung_Mosca, Watrous}.

If $rb\in B_k$ for all $r\in\tilde{R}$ and all $b\in\tilde{B}_k$, we stop since $I=B_k$ as shown by the lemma below. If $rb\not\in B_k$ for some $r\in\tilde{R}$ and some $b\in\tilde{B}_k$, we set $\tilde{B}_{k+1} = \tilde{B_k} \cup \{rb\}$ and continue with this new generating set. By Lagrange's theorem, we know that $\frac{|B_{k+1}|}{|B_k|}\ge 2$ since $B_k$ is a strict subgroup of $B_{k+1}$.  

\begin{lemma}
Let $I$ be a left ideal generated by $\{i_1,\ldots,i_m\}$ in a finite ring $R$.  Let $\tilde{B}_k$ be a subset of $I$ containing $\{i_1,\ldots,i_m\}$. The set of ring elements $B_k$ additively generated by $\tilde{B}_k$ is equal to $I$ if and only if $r b\in B_k$ $\forall r \in \tilde{R}$ and $\forall b \in\tilde{B}_k$.
\end{lemma}
\begin{proof}
First, if $r b\in B_k$ for all $b\in \tilde{B}_k$ (for an arbitrary but fixed $r\in\tilde{R})$, then $r B_k \subseteq B_k$.  Second, if $rB_k \subseteq B_k$ for all $r \in \tilde{R}$ then, because $\tilde{R}$ is a generating set for $R$, $rB_k \subseteq B_k$ for all $r \in R$. Thus, $B_k$ is a left ideal in $R$. By construction, $B_k$ contains $i_1,\ldots,i_m$. By the definition of generators, $I$ is the smallest left ideal in $R$ containing $i_1,\ldots,i_m$. $B_k$ is also contained in
$I$. Thus $B_k = I$. The converse follows immediately from the fact that $I$ is a left ideal.
\end{proof}

In the above procedure, the number of trials needed to obtain each additive generator is at most $|\tilde{R}| \cdot |\tilde{B}_k|$ since in the worst case we have to test $r b$ for all pairs $r\in\tilde{R}$ and $b\in\tilde{B}_k$. Furthermore, every time we add another generator, we increase the size of the generated group by at least a factor of two. Thus, we need to perform the above iteration at most $\log_2|I|$ times. The resulting overall cost is $\mathrm{poly}(\log |R|)$.

We can also in polynomial time obtain expressions for the elements of this set in terms of the original generators for $I$ by recursively composing the expressions we obtained at each step for $i$ in terms of the preceding generators $B_k$.

\section{Decomposition into additive generators}

Once we have a set $\tilde{B}_k$ of elements that generate $I$ as an Abelian group, we can efficiently find an i.~f. generating set for $(I,+)$, as well as expressions for the i.~f. generators as linear combinations of $\tilde{B}_k$ using the techniques of \cite{Cheung_Mosca, Watrous}. These techniques also efficiently yield the additive orders of the i.~f. generators.

After finding an i.~f. generating set for $(I,+)$, one would like to have a procedure to take a given element $i \in I$ and decompose it as a linear combination of these generators. Note that $i$ is given as an arbitrary bit string from the encoding $\eta$, so initially we know nothing about $i$. We can efficiently perform this decomposition as described below.

Let $G = \mathbb{Z}_{s_1} \times \mathbb{Z}_{s_2} \times \ldots \times \mathbb{Z}_{s_\ell} \times \mathbb{Z}_{s}$, where $s_1,\ldots,s_\ell$ are the orders of the i.~f. generators $h_1,\ldots,h_\ell$ and $s$ is the order of $i$. Let 
\[
  f(n_1,n_2,\ldots,n_\ell,m) = \eta \left( \sum_{j=1}^{\ell} n_j h_j + mi \right).
\]
This function hides the cyclic subgroup of $G$ generated by
\[
  (n_1(i),n_2(i),\ldots,n_\ell(i),-1),
\]
where $n_1(i),\ldots,n_\ell(i)$ is the decomposition of $i$ in terms of the i.~f. generators:
\[
  i = \sum_{j=1}^{\ell} n_j(i) h_j.
\]
Using the polynomial time quantum algorithm for the Abelian hidden subgroup problem \cite{Nielsen_Chuang}, we thus recover this decomposition.

The multiplication in $I$ can be fully specified by the tensor $M_{ij}^k$ defined by
\[
  h_1 h_j = \sum_{k=1}^\ell M_{ij}^k h_k.
\]
We can compute all $l^3$ of the entries of $M_{ij}^k$ by taking each pair $h_i,h_j$, using the multiplication oracle to find
the bit string encoding their product, and then using the Abelian hidden subgroup algorithm to decompose the element represented by the resulting bit string, as described above. Together, the i.~f. generators for $I$, their orders, and the multiplication tensor are called a basis representation for $I$. The previous work of Arvind \emph{et al.}  shows how to efficiently quantum compute a basis representation in the special case that $I$ is the entire ring $R$ \cite{Arvind}. The best existing classical algorithm for this problem requires order $|R|$ queries \cite{Zumbragel}.

\section{Application to ideal problems}

Given a basis representation for an ideal $I$ it is straightforward to construct a uniform superposition $\ket{I}$ over all elements of $I$. By constructing the superpositions $\ket{I}$ and $\ket{J}$ for two ideals $I$ and $J$ we can determine whether $I=J$ using the swap test. By Lagrange's theorem, if $I \neq J$ then $\braket{I}{J} \leq 1/2$. Thus we need only use
$O(\log(\epsilon))$ swap tests to ensure that the chance of falsely concluding $I=J$ is at most $\epsilon$.

A reversible circuit for addition performs the unitary transformation $U_+ \ket{a} \ket{b} = \ket{a} \ket{a+b}$. Thus, $U_+ \ket{a}\ket{J} = U_+ \ket{a}\ket{a+J}$, where $\ket{a+J}$ is a superposition over the coset $a+J$. Thus, after constructing $\ket{I}$ and being given a ring element $r$, we can use the addition black box to construct the coset state $\ket{r+I}$. If $r \in I$ then the inner product of these states is one, and otherwise it is zero. Thus, the swap test on $\ket{I}$ and $\ket{r+I}$ tells us whether $r \in I$. Given $r \in R$, let $Rr$ be the left ideal in $R$ generated by $r$. $Rr = R$ if and only if $r$ is a unit. If $Rr \neq R$ then $Rr$ contains at most half the elements of $R$. Thus one can determine whether a given $r \in R$ is a unit by constructing $\ket{Rr}$ and $\ket{R}$ and comparing them using the swap test. If $r$ is a unit, then we can find its inverse using the quantum order finding algorithm \cite{Shor_factoring}. If $r^c = \id$ then $r^{-1} = r^{c-1}$.

Suppose $r$ is contained in the ideal $I$. To obtain an explicit construction for $r$ in terms of the generators of $I$, we can first obtain a basis representation for $I$. We can obtain an expression for $r$ as a linear combination of the basis for $I$ by solving the Abelian hidden subgroup problem. From the algorithm for obtaining a basis representation for $I$ we also obtain expressions for the basis elements in terms of the original generators of $I$. Thus one can efficiently convert the expression for $r$ as a linear combination of the basis representation for $I$ into an expression for $r$ in terms of the original generators for $I$.

Suppose we are given generating sets for two ideals $I$ and $J$. We wish to find a basis for $I \cap J$. By techniques described above, we can create the superposition $\ket{J}$ over all elements of $J$, and we can find a basis representation for $I$. If $a \in J$ then $\braket{a+J}{J} = 1$. Otherwise $\braket{a+J}{J} = 0$. Hence applying addition to the state $J$ is an operation that ``hides'' the subgroup $(I \cap J,+)$ of the group $(I,+)$ of inputs. Thus, one can use the quantum algorithms for the Abelian hidden subgroup problem \cite{Nielsen_Chuang} to find a set of generators for $(I \cap J,+)$. From this we easily extract a basis representation. (Typically in a hidden subgroup problem one is given a blackbox that maps group elements to classical bit strings. This map is constant and distinct on cosets of the hidden subgroup. However, examining the algorithm of \cite{Nielsen_Chuang}, one sees that it works just the same if the blackbox maps the different cosets to any set of orthogonal states, the classical bit string states being just a special case.)

If $I$ and $J$ are two ideals in $R$, one defines $(I:J) = \{x \in R| xJ \subseteq I\}$. $(I:J)$ is an ideal, and is called an ideal quotient or a colon ideal. $(I:J)$ is a subgroup of $(R,+)$. Let $U$ be the unitary transformation defined by $U\ket{x} \ket{y_1} \ldots \ket{y_m} = \ket{x} \ket{xj_1+y_1} \ldots \ket{xj_m+y_m}$ for all $x,y_1,\ldots,y_m \in R$. Given quantum black boxes for arithmetic on $R$, $U$ can be efficiently implemented by a quantum circuit. The states $\ket{xj_1 + I} \ldots \ket{xj_m+I}$ and $\ket{yj_1 + I} \ldots \ket{yj_m+I}$ are identical if $x$ and $y$ belong to the same coset of
$(I:J)$ in $(R,+)$ and are orthogonal if $x$ and $y$ come from different cosets. Thus, we can efficiently find an additive generating set for $(I:J)$ by solving the Abelian hidden subgroup problem using $U$ to hide $(I:J)$.

The left annihilator $A_S$ of $S = \{s_1,\ldots,s_n\} \subseteq R$ is defined as $A_S = \{x \in R|xs_1 = 0, \ldots, xs_n = 0 \}$. $A_S$ forms a subgroup of $(R,+)$. The function on $R$ given by $f_S(x) = (xs_1,\ldots,xs_n)$ hides this subgroup. Thus, after finding an i.~f. generating set for $R$ one can use the quantum algorithm for the Abelian hidden subgroup problem to find generators for any annihilator provided $S$ is at most polynomially large. The same method will work if $S$ is given by a polynomially large set of additive generators. 

Given generators for an ideal $I$ in a finite ring, we can find the order of $I$, by finding an i.~f. generating set for it and taking the product of the orders of the generators. Finding the order of a ring is a special case, as any ring is an ideal in itself.

Suppose we are given a black-box implementing a homomorphism $\rho:R \to R'$ between two rings. Determining whether $\rho$ is injective is an Abelian hidden subgroup problem, where the kernel of $\rho$ is the hidden subgroup in $(R,+)$. $\rho$ is injective if and only if its kernel is $\{0\}$. We can efficiently find generators for the kernel of $\rho$ by finding an i.~f. generating set for $R$, and then solving the Abelian hidden subgroup problem. To determine whether $\rho$ is surjective, we first compute the order of $R'$.  Similarly, the image of $\rho$ is a ring. If $R$ is generated by $\{r_1,\ldots,r_n\}$ then $R'$ is generated by $\{\rho(r_1),\ldots,\rho(r_n)\}$. After querying the homomorphism black-box to obtain the generators $\{\rho(r_1),\ldots,\rho(r_n)\}$ we can compute the order of the ring they generate ($R'$) as described in the preceding paragraph. $\rho$ is surjective if and only if the order of the image of $\rho$ equals the order of $R'$.

Suppose we wish to solve a linear equation $ax=b$ over $R$. To do this we find an i.~f. generating set $\{h_1,\ldots,h_\ell\}$ for $R$, and decompose $a$ and $b$ in terms of these generators
\[
  \begin{array}{lccr}
  a = \sum_{i=1}^{\ell} a_i h_i & \quad & \quad &
  b = \sum_{i=1}^{\ell} b_i h_i.
  \end{array}
\]
Let
\[
  A_{ij} = \sum_k a_k M_{kj}^i
\]
where $M_{kj}^i$ is the multiplication tensor from the basis representation. Parametrize $x$ as $x = \sum_{i=1}^{\ell} x_i h_i$ for integers $x_1,\ldots,x_\ell$. Then, in an i.~f. generating set, $ax=b$ if and only if
\begin{equation}
  \label{mods}
  \sum_{j=1}^\ell A_{ij} x_j \equiv b_i \mod s_i,
\end{equation}
for each $i=1,2,\ldots,\ell$. (Here $s_i$ is the additive order of $h_i$.) We can introduce additional integer unknowns $k_1,\ldots,k_\ell$ and rewrite this as a system of linear Diophantine equations:
\begin{equation}
  \label{system}
  \sum_{j=1}^\ell A_{ij} x_j + k_i s_i = b_i,\quad i=1,2,\ldots,\ell.
\end{equation}
A solution to a system of $m$ Diophantine equations in $n$ variables can be found in $\mathrm{poly}(n,m)$ time using the classical algorithms of \cite{Chou}. Thus we can classically find an integer solution to equation \ref{system}, which has $\ell$ equations and $2\ell$ unknowns, in $\mathrm{poly}(\ell)$ time. Equation \ref{system} is undertedermined because the original system of equations \ref{mods} is modular.

By a similar technique, we can find the identity in $R$. Again suppose we have computed a basis representation for $R$. Since the basis representation has the following property,
\[
  n_1 h_1+\ldots+n_\ell h_\ell=h_{\alpha} \Rightarrow 
  n_{\beta}=\delta_{\alpha\beta} \quad 1\leq \beta\leq \ell
\]
where $n_i\in \mathbb{Z}_{s_i}$, an element $r = \sum_{i=1}^{\ell} r_i h_i$ is the identity if and only if
\[
  \sum_{i=1}^{\ell} r_i M_{ij}^k \equiv \delta_{jk} \mod s_k
\]
for all $j,k=1,2,\ldots,\ell$. This is again a system of linear modular equations, which we can convert to a system of linear Diophantine equations that we solve in polynomial time using \cite{Chou}. Note that the quantum algorithm of \cite{Arvind_ident} solves a very different problem although the authors refer to it as identity testing.

In a black box ring, finding the additive identity is also nontrivial. Because all ring elements have additive inverses, we can choose any $r \in R$, find its order $c$ using the quantum order finding algorithm \cite{Shor_factoring}, find the additive inverse of $r$ by computing $(c-1)r$, and find the additive identity by computing $cr$. The computation of $cr$ and $(c-1)r$ requires $O(\log_2 c)$ queries to $f_+$. 

We now show how to efficiently determine whether a given two-sided ideal $I$ is prime. Recall that an ideal $I$ is prime if
$ab\in I$ implies that $a\in I$ or $b\in I$ for all $a,b\in R$, which is equivalent to the fact that the quotient ring $S=R/I$ does not have any zero-divisors. This already implies that $S$ is a division ring (i.e., each non-zero element has a multiplicative inverse) since $S$ is finite.  Wedderburn's theorem shows that all finite division rings are finite fields \cite{Lidl_Niederreiter}. $R/I$ a field implies $I$ is maximal, thus $I$ is prime implies $I$ is maximal. The converse is also true. 

Let $S^*$ denote the group of units of the quotient ring $S$.  We choose an element $r$ uniformly at random in $R$. With probability at least $1/2$ we have $r\not\in I$.  Once we obtain such $r$ we determine the size of the (additively generated) cyclic subgroup $\langle \bar{r} \rangle$ of $S$, where $\bar{r}$ denotes the image of $r$ in $S$ under the canonical projection. This can be done by applying Shor's period finding algorithm to the state $(1/\sqrt{q}) \sum_{x=0}^ q |x\rangle |xr + I\rangle$) where $q$ is a power of $2$ with $|S|^2 < q\le 2|S|^2$.  This state can be prepared efficiently.

If $S$ is a field, then with probability at least $\varphi(|S|-1)/|S| \ge \Omega(1/\log |S|)$ we have $\langle \bar{r} \rangle = S^*$ where $\varphi$ denotes Euler's totient function. This follows from the fact that the group of units $\mathbb{F}_d^*$ of an arbitrary finite field $\mathbb{F}_d$ with $d$ element is cyclic of order $d-1$ and $\varphi(m)/m = \Omega(1/\log m)$ for integers $m$ \cite{Hardy}. If $S$ is not a field, then $S^*$ cannot have order $|S|-1$ (otherwise every non-zero element would have a multiplicative inverse, implying that $S$ is a field). If we find that $S$ is a field then we know $I$
is prime, otherwise $I$ is not prime. The above procedure for determining whether the quotient ring $S$ is a field can be applied to any finite blackbox ring, offering a simpler alternative to the quantum algorithm in \cite{Arvind}.

Our quantum algorithms for rings $R$ also extend to $R$-modules. Beyond this, we conjecture that our quantum algorithms
apply to any category posessing a faithful functor to the category of Abelian groups. 

It would be interesting to find efficient quantum algorithms for deciding whether a given ideal $I$ is principal and computing the group of units $R^*$ of $R$. The quantum algorithms in \cite{Cheung_Mosca, Watrous} make it possible to determine the structure of any finite abelian black-box group according to the structure theorem. So, the question arises naturally whether a similar quantum algorithm exists for decomposing finite black-box rings. More precisely, is it
possible to efficiently learn the structure of a finite black-box ring according to a structure theorem in ring theory such as the Wedderburn-Artin theorem \cite{McDonald}?

It would be worthwhile to investigate whether the above algorithms extend to the case of infinite rings.  It is not obvious that we can consider arbitrary infinite rings. However, it seems likely that the above algorithms could be extended to a black-box ring $R$ which is endowed with a grading by Abelian groups $R_0,R_1,R_2,\ldots$ and each component $R_g$ is finite.  Additionally, we would need a promise, making it possible to do all the computations in a component $R_g$ for
some $g$. For example, such a situation occurs for polynomial rings over a finite field when the number of indeterminates is fixed. The complexity of the algorithms would then depend on the growth of the Hilbert function, which measures the dimension of the graded components $R_g$ as $R_0$-modules.

\section{Alternative Approach}
\label{sec:alt}

Rather than directly developing quantum algorithms for the solution of problems about ideals, it is also possible to solve these problems by applying the quantum algorithm of \cite{Arvind} to obtain a basis representation for the underlying ring $R$, decomposing all ring elements specifying the problem input as integer linear combinations of this basis, and then reling on efficient classical algorithms for solving systems of linear Diophantine equations. In other words, rather than a direct quantum attack, as described in prior sections, the solution of problems regarding ideals in black box rings can be conceptualized in terms of a black box for Abelian hidden subgroup problems plus classical algorithms for solving systems of linear Diophantine equations.

Here, we illustrate this principle for the problem of deciding whether a given ring element is in a given ideal. The solutions to the other problems regarding ideals are similar in spirit and can be obtained using known techniques for manipulating systems of linear Diophantine equations such as those described in section II of \cite{Schrijver}.

Suppose we are given access to a black box ring $R$, a ring element $r \in R$ and a generating set $\tilde{I}$ for an ideal $I$ in $R$. We wish to determine whether $r \in I$. Rather than directly attacking this problem with the quantum algorithm described above, we can instead start by using the quantum algorithm of \cite{Arvind} to obtain a basis representation for $R$. A basis representation for $R$ consists of:
\begin{itemize}
  \item an i.~f. generating set $\{r_1,\ldots, r_\ell\}$ for $R$
  \item the additive orders of the elements $\{r_1,\ldots, r_\ell \}$
  \item the multiplication tensor $M_{ij}^k$ such that $r_i r_j = \sum_k M_{ij}^k r_k$.
\end{itemize}

An instance of the problem of deciding whether a given ring element $r$ is contained in an ideal is specified by $r$ and the set $\tilde{I}$ of ring elements that generate $I$. In the black box setting these are specified by their labels, which are arbitrary bit strings. With a quantum computer, we can efficiently decompose these ring elements as integer linear combinations of the additive generators $\{r_1,\ldots, r_\ell\}$ by solving Abelian hidden subgroup problems.

After the quantum preprocessing steps to obtain an i.~f. generating set for $R$ and decompositions of all ring elements defining the problem instance, we can solve the problem using only classical algorithms for solving systems of linear Diophantine equations.

The first step is to construct an i.~f. generating set for $I$. This can be done by the same process described in section \ref{sec:if}, except that to test whether a given element of $\tilde{R}$ is contained in $B_k$, we use classical algorithms for deciding whether it is in the integer linear span of the i.~f. generators for $B_k$ rather than relying on Hadamard tests.

Specifically, let $\{ b_1, \ldots, b_{n(k)} \}$ denote the elements of $\widetilde{B}_k$ and let $b_i = \sum_j \beta_{i,j} h_j$ with $\beta_{i,j} \in \mathbb{Z}$ denote their decompositions in terms of the i.~f. generators $\{h_1,\ldots,h_\ell\}$ for $R$. Similarly, for a generator $r' \in \tilde{R}$, let $r' = \sum_j \rho_j' h_j$ denote its decomposition in terms of the i.~f. generators of $R$. Then, $r' \in B_k$ if and only if there exists $z_1,\ldots,z_{n(k)} \in \mathbb{Z}$ such that
\begin{eqnarray*}
  z_1 \beta_{1,1} + \ldots + z_{n(k)} \beta_{n(k),1} & \equiv & \rho_1' \mod s_1 \\
  & \vdots & \\
  z_1 \beta_{1,\ell} + \ldots + z_{n(k)} \beta_{n(k),\ell} & \equiv & \rho_1' \mod s_\ell.
\end{eqnarray*}
This system of linear congruences can be converted to an equivalent system of linear Diophantine equations simply by introducing slack variables $\{v_1, \ldots, v_\ell\}$:
\begin{eqnarray*}
  z_1 \beta_{1,1} + \ldots + z_{n(k)} \beta_{n(k),1} + v_1 s_1 & = & \rho_1' \\
  & \vdots & \\
  z_1 \beta_{1,\ell} + \ldots + z_{n(k)} \beta_{n(k),\ell} + v_\ell s_\ell & = & \rho_1'.
\end{eqnarray*}
Determining whether such a linear system has a solution over the integers (and, if so, finding it) can be done classically in polynomial time by reduction to Hermite normal form, as discussed in e.g. \cite{Frumkin, Hafner, Schrijver}.

Once this process converges upon a complete set of additive generators for $I$, one can then test whether $r$ is in $I$ by using these same classical algorithms to determine whether its decomposition $\rho_1,\ldots, \rho_\ell$ is within the integer linear span of the decompositions of these generators, modulo the appropriate additive orders.

\medskip
{\bf Acknowledgements}. This work was mainly carried out in 2009 while P.W.~was with the School of Electrical Engineering and Computer Science, University of Central Florida (UCF), Orlando, Florida, S.P.J.~with the Institute for Quantum Information, Caltech, Pasadena, and H.A.~and J.P.B.~with the Department of Mathematics, UCF. The updated version fixes a missing analysis of a certain subcase that was pointed out by Ethan Martin in 2023 \cite{Martin23}.  The authors would like to thank him for his valuable feedback. 

P.W.\ and H.A.\ gratefully acknowledge the support of NSF grants CCF-0726771 and CCF-0746600. S.J.\ gratefully acknowledges support from the Sherman Fairchild foundation and the NSF under grant PHY-0803371. 

\medskip

\bibliography{rings}

\end{document}